\documentclass[10pt, journal]{IEEEtran}
\usepackage{cite}
\usepackage{amsmath,amssymb,amsfonts}
\usepackage{graphicx}
\usepackage{algorithm}
\usepackage[hidelinks]{hyperref}
\usepackage{textcomp}

\usepackage[utf8]{inputenc} 
\usepackage[T1]{fontenc}    
\usepackage{amsmath,amssymb,amsfonts}
\usepackage{amsthm}
\usepackage{algpseudocode}
\usepackage{todonotes}
\usepackage{multirow}
\usepackage{url}
\usepackage{cite}
\usepackage{textcomp}
\usepackage{xcolor}
\usepackage{optidef}
\usepackage{arydshln}
\usepackage{balance}
\usepackage{tikz}
\usepackage{circuitikz}
\usetikzlibrary{arrows.meta}
\tikzstyle{block} = [draw, rectangle, 
    minimum height=3em, minimum width=3em]
\tikzstyle{sum} = [draw, circle, node distance=1cm]
\tikzstyle{input} = [coordinate]
\tikzstyle{output} = [coordinate]
\tikzstyle{pinstyle} = [pin edge={to-,thin,black}]

\graphicspath{{./figures/}}

\usepackage{bbm}

\newtheorem{theorem}{Theorem}

\newtheorem{remark}{Remark}
\newtheorem{definition}{Definition}
\newtheorem{exmp}{Example}

\newcommand{\BBM}{\begin{bmatrix}}
\newcommand{\EBM}{\end{bmatrix}}
\newcommand{\BEQ}{\begin{equation}}
\newcommand{\EEQ}{\end{equation}}
\newcommand{\BIT}{\begin{itemize}}
\newcommand{\EIT}{\end{itemize}}

\newcommand{\complex}{\mathbb{C}}

\DeclareMathOperator{\rank}{rank}

\DeclareMathOperator{\Expect}{\mathbb{E}}

\newcommand{\E}[2][]{\Expect_{#1}\!\left[\,#2\,\right]}   

\begin{document}
\title{Informativity and Identifiability for Identification of Networks of Dynamical Systems}
\author{Anders Hansson, \IEEEmembership{Senior Member, IEEE}, João Victor Galvão da Mata, and Martin S.~Andersen
\thanks{This work was supported by ELLIIT, and by the Novo Nordisk Foundation under grant number NNF20OC0061894.}
\thanks{Anders Hansson is with the Department of Electrical Engineering, Linköping University (e-mail: anders.g.hansson@liu.se).}
\thanks{João Victor Galvão da Mata and Martin S.~Andersen are with the Department of Applied Mathematics and Computer Science, Technical University of Denmark (e-mail: jogal@dtu.dk; mskan@dtu.dk).}}
\maketitle
\begin{abstract}
In this paper, we show how informativity and identifiability for networks of dynamical systems can be investigated using Gröbner bases. We provide a sufficient condition for generic informativity in terms of
positive definiteness of the spectrum of external signals and full generic row rank of the transfer function relating the external signals to the inputs of the predictor. Moreover, we show how generic local network identifiability in the algebraic sense can
be investigated by computing the dimension of the
generic fiber associated with the closed loop transfer function from external
measurable signals to the measured outputs.  
\end{abstract}
\section{Introduction}
\IEEEPARstart{S}{ystem} identification for networks of dynamical systems is an active research 
field with a 
long history. Some early references 
are \cite{VANDENHOF19931523,VANDENHOF20132994,hen+gev+baz19}. Some of the work has been focusing on identifying
the full network, e.g. \cite{FONKEN2022110295,hen+gev+baz19,WEERTS2018256}, while other work has been focusing on identifying parts of the network,
a so-called {\it sub-network}, e.g. 
\cite{VANDENHOF20132994,Dankers_predictor_input,Materassi2019SignalSF,7402842,Linder03042017,9247487,9661422}.
The latter is an example of system 
identification under feedback, a topic that also has a long history, see e.g. \cite{FORSSELL19991215}.

Many system identification methods are based on a so-called {\it predictor}, and the system identification 
methods using predictors are called {\it prediction error methods}. 
Two fundamental questions related to these methods are what is called {\it informativity} and {\it identifiability}. 
The first question is about if the signals used for system identification are such that the predictor
can be uniquely determined. The second question is whether the network topology is such that the open loop transfer
functions can be uniquely determined from the predictor. This paper will focus on these questions under the 
assumption that a stable predictor exists. We do not care if this predictor is linear or nonlinear in the 
open loop transfer functions. We will investigate these questions both for the full network case and 
for the sub-network case. We will see that we can handle more general network configurations than what has been 
possible hitherto in the literature. Specifically, we will consider the {\it partial measurement case}, i.e. not 
all of the node signals in the (sub)-network are in the measured outputs. Moreover, we will allow the measurement equation 
to contain transfer functions to be identified, and we will also allow for direct feed-through from the inputs to the
measured outputs. In a seamless way, we also allow for some transfer functions to be known and for rational 
constraints relating 
transfer functions to one another. 

Regarding informativity, we focus on sufficient conditions for generic informativity in terms of
positive definiteness of the spectrum of external signals and full generic row rank of the transfer function relating the
external signals to the inputs of the predictor. This has previously been studied in e.g. 
\cite{VANDENHOF20132994,BOMBOIS2023110742}. We specifically generalize this work to the partial measurement case,  and to some of the transfer functions to be known.
We give a graph-theoretic characterization of the generic rank in the free-entry setting,
and an algebraic geometry characterization that also accommodates known transfer functions and constraints.

Moreover, we will derive conditions for generic local
network identifiability in the algebraic sense in terms of the dimensions of the admissible domain
and the closure of its image. This is a generalization of the work in
\cite{hen+gev+baz19,SHI2022110093} and complements the work in 
\cite{leg25}. Specifically, we are able to consider some of the transfer functions to be known, and
constraints on transfer functions. We are also able to allow for the measurement equation 
to contain transfer functions to be identified.
\section{Model of Network}
We consider a linear dynamic network model defined by 
\BEQ\label{eqn:PMD-model}
\begin{aligned}
P(q)w_k&=Q(q)u_k\\
z_k&=R(q)w_k+S(q)u_k
\end{aligned}
\EEQ
where the quadruple $(P,Q,R,S)$ are transfer function matrices
in the forward shift operator $q$. The dimensions of these transfer function matrices are $n\times n$, 
$n\times m$, $p\times n$, and $p\times m$, respectively. This also defines the dimensions of all the 
signals, where $u_k$ is referred to as the input signal, $z_k$ as the output signal, and $w_k$ as the node signal. We assume that $P$ is invertible for almost all $q$
such that the closed loop system is {\it well-posed}.
We may sometimes replace $q$ with the complex variable $z$. 
We  assume that some of the elements of the transfer functions
we have defined are dependent on some parameter $\theta$, which 
is going to be identified. 
If we want to emphasize the parameter dependence we may e.g.
write $P(q;\theta)$ for transfer functions and $y_k(\theta)$ for signals. We will often drop parameter dependence
and dependence on $q$. 
\section{Innovation form and Predictor}\label{sec:innovation}
We divide $u_k$ in a part $r_k$ that we know or measure,  and another part $e_k$, which we refer to as noise. 
We write $u_k=(r_k,e_k)$, where $r_k$ has $m_r$ components and $e_k$
has $m_e$ components. We also partition $Q(q)$ and $S(q)$ conformably as 
$$Q(q)=\BBM Q_r(q)&Q_e(q)\EBM,\quad S(q)=\BBM S_r(q)&S_e(q)\EBM.$$
We assume that 
$z_k$ can be partitioned as $z_k=(y_k,r_k)$, where all of $z_k$ is a signal that
we have full information about. Here $y_k$ has $p_y$ components, and
$r_k$ has $m_r$ components, so $p=p_y+m_r$.
We conformally write
\begin{align*}\BBM R(q)&S(q)\EBM&=\BBM R(q)&S_r(q)&S_e(q)\EBM\\&=
\BBM R_y(q)&S_{yr}(q)&S_{ye}(q)\\0&I&0\EBM
\end{align*}
for appropriately defined transfer function matrices. From this, we may conclude
that the {\it closed loop system} from $u_k$ to $y_k$ may be written 
\BEQ
y_k=G_c(q)r_k+H_c(q)e_k\label{eqn:closed}
\EEQ
where
\BEQ\label{eqn:closed-tf}
\begin{aligned}
\BBM G_c(q)&H_c(q)\EBM&=R_y(q)P(q)^{-1}\BBM Q_r(q)&Q_e(q)\EBM\\&\quad+
\BBM S_{yr}(q)&S_{ye}(q)\EBM.
\end{aligned}
\EEQ
System identification is very often based on a {\it predictor}, and to this
end we define the so-called {\it innovation form}, i.e. we write
\begin{align}
\hat y_k&=G_c(q) r_k+G_o(q)\epsilon_k\label{eqn:innovation}\\
y_k&=\hat y_k+\epsilon_k\label{eqn:error}
\end{align}
for some transfer function matrix 
$G_o(q)$, where the {\it innovations} $\epsilon_k$ are zero mean
equally distributed and independent
random variables. We remark that obtaining the innovation form is by no means
trivial and often involves solving an algebraic Riccati equation, e.g.
\cite{hansson2025identifiabilitymaximumlikelihoodestimation}. However, we will only need the existence of the innovation form in
this work, and no explicit knowledge of the transfer function $G_o(q)$ is
needed.

From the above equations we realize that we may write
\BEQ
\hat y_k= W(q)z_k\label{eqn:pred}
\EEQ
where 
$$W(q)=\BBM W_y(q)&W_r(q)\EBM =(I+G_o(q))^{-1}\BBM G_o(q)&G_c(q)\EBM.$$ 
This is called the {\it predictor} for $y_k$. 
Here we need to make the assumption that $W(q)$ has all poles strictly inside
the unit circle. This is not a priori given, and has to be verified for each 
given example. 
It follows using simple manipulations of the above equations 
that  $ G_c(q)$ can be recovered from the predictor as
\BEQ\label{eqn:predtogc}
 G_c(q)=(I- W_y(q))^{-1} W_r(q).
\EEQ
We will now revisit some special cases of the above network formulation that
are popular in the literature. Most existing work assume that $P(q)=I-G(q)$ for
some transfer function matrix $G(q)$ which is assumed to have a zero diagonal. 
Most existing work also only consider $S(q)=0$.
Moreover, it is common to have $R_y(q)=I$ implying $y_k=w_k$, which
is referred to as the {\it full measurement} case. The {\it partial 
measurement} case is when $R_y(q)$ is a matrix having rows equal to 
standard basis vectors, which means that $y_k$ contains components of $w_k$.
It is also very common that $Q_r(q)$ and $Q_e(q)$ are diagonal matrices. 
Moreover, it is not uncommon that they are independent of $q$. This specifically
goes for $Q_r(q)$. Notice that we do not make any assumptions similar to the special cases above. 
\section{Informativity}
We will now investigate what is called informativity for estimation of the parameter $\theta$. 
To this end, the following definition is useful:
\begin{definition}
The signal $z_k$ is said to be {\it informative enough} if for any two parameters 
$\theta_1$ and $\theta_2$ such that the closed loop system is stable, it holds that 
\BEQ\label{eqn:informative}
\E{\left(\hat y_{k}(\theta_1)-\hat y_{k}(\theta_2)\right)
\left(\hat y_{k}(\theta_1)-\hat y_{k}(\theta_2)\right)^T}=0
\EEQ
implies that $ W\left(e^{i\omega};\theta_1\right)= W\left(e^{i\omega};\theta_2\right)$ for almost  
all $\omega$. 
\end{definition}
It holds that 
$$\hat y_k(\theta_1)
-\hat y_k(\theta_2)=\Delta W(q)z_k$$
where $\Delta W(q)= W(q;\theta_1)- W(q;\theta_2)$. 
From Parseval's formula, it follows that 
\eqref{eqn:informative} is equivalent to 
$$\Delta  W\left(e^{i\omega}\right)\Phi_z(\omega)
\Delta  W\left(e^{-i\omega}\right)^T=0$$
for almost all $\omega$, where $\Phi_z$ is the spectrum of $z_k$. Hence, a sufficient condition for the 
above equality to imply that $ W\left(e^{i\omega};\theta_1\right)= W\left(e^{i\omega};\theta_2\right)$ for almost  
all $\omega$, is that $\Phi_z(\omega)$ is positive definite for almost all 
$\omega$.

We have that the closed-loop transfer
function from $u_k=(r_k,e_k)$ to 
$z_k$ is given by
$$\Pi(q)=R(q)P(q)^{-1}Q(q)+S(q)$$
and that the spectrum of $z_k$ is 
positive definite if $\Pi\left(e^{i\omega};\theta_0\right)$ is full row rank for almost all $\omega$ and the spectrum
of $u_k$ is positive definite. 
This condition depends on the true system described
by $\theta_0$, and hence it is a difficult condition to verify. 
Let
\BEQ
M(q)=\BBM P(q)&Q(q)\\-R(q)&S(q)\EBM.
\EEQ
Similarly as
in \cite{BOMBOIS2023110742}, we will instead only investigate the so-called
{\it generic} rank defined as 
$$\rank_{g} \Pi=\max_{M}
\rank\Pi.$$
Here we consider $\Pi$ to just be a function of the free entries of the 
matrix $M$.\footnote{When no ambiguity arises we will often just refer to  $M$ even if we mean  the 
free entries of $M$.}  The only property we need regarding these free entries when we investigate generic rank
is that they are algebraically independent\footnote{Algebraically independent is the formal way of saying that the entries are free. } commuting indeterminates over the coefficient field.
The coefficient field is the field of 
rational functions of transfer functions. This is a field with the standard definition of addition and multiplication.
We will later on see that the generic rank is equal to the rank for almost all
$M$.
Since $\Pi$ is a rational function in $z$, it follows that 
$$\rank_g\Pi(z)=\max_z \rank_g\Pi(z)$$
for almost all $z\in\complex$, and hence
$$\rank \Pi(z)=\max_{z}\rank_g \Pi(z)$$
for almost all $M$ and $z$.
Using this result, it is possible to give a sufficient condition for informativity
for almost all admissible true networks.
\begin{theorem}
Under the standing assumptions on the network and the candidate predictors,
suppose that the spectrum of $u_k=(r_k,e_k)$ is positive definite for almost all
frequencies and that $\Pi$ has full generic row rank. Then, for almost all
admissible choices of the true network's free transfer functions, the resulting
signal $z_k$ is informative enough.
\end{theorem}
\begin{proof}
For almost all admissible choices of the true network's free transfer functions,
full generic row rank implies that $\Pi(e^{i\omega};\theta_0)$ has full row rank
for almost all $\omega$. Since
$$\Phi_z(\omega)=\Pi(e^{i\omega};\theta_0)\Phi_u(\omega)
\Pi(e^{i\omega};\theta_0)^*,$$
where $^*$ denotes the conjugate transpose, the spectrum of $z_k$ is then
positive definite for almost all frequencies. Informativity follows from the
preceding argument.
\end{proof}
Here, ``almost all'' has the same meaning as in the preceding generic-rank
discussion. This is a generic guarantee: informativity may fail for exceptional
true networks, for example when exact cancellations reduce the rank of $\Pi$.
The qualification concerns the true network generating the data; for each true
network covered by the theorem, the implication in the definition of
informativity still holds for every pair of admissible candidate predictors.
The generic rank is understood relative to the model class under consideration,
respecting any known transfer functions or imposed constraints.
\begin{remark}\label{rem:predictor-reduction}
If some components of $u_k$ are set to zero, the corresponding columns of
$M(q)$ may be removed. If some predictor columns are fixed across all candidate
models, the measurement rows of $M(q)$ corresponding to their predictor inputs
may be omitted when testing informativity for the remaining columns. Merely
being uninterested in unknown predictor columns does not justify this removal,
since changes in those columns may cancel changes in the columns of interest.
More generally, known relationships between predictor columns may allow the
predictor to be expressed using fewer, known combinations of the measured
signals. Informativity can then be investigated for this reduced predictor-input
vector.
\end{remark}
What has been presented above
is a generalization of the results in \cite{BOMBOIS2023110742} for a much 
more general network configuration.
We will now 
give conditions on the generic rank of $\Pi$ in terms of both 
a graph and in terms of so-called {\it Gröbner basis}. 
\subsection{Graphical Condition}
Since we are investigating generic properties, we drop the dependence on
the forward shift operator $q$ or the complex variable $z$.
For the graph construction, we assume that the diagonal entries of $P$ are
structurally nonzero, as for $sI-A$ in \cite{wou91}. Since $P$ is generically
invertible, this can always be achieved by a row permutation. Applying the same
permutation to $Q$ leaves $RP^{-1}Q+S$ unchanged. In this subsection, $P$, $Q$,
and $M$ refer to the matrices after any such permutation.

The following equality is easily proven by multiplying together the matrices:
\begin{align*}
&\BBM I&0\\RP^{-1}&I\EBM M\BBM I&-P^{-1}Q\\0&I\EBM\\&=
\BBM P&0\\0& RP^{-1}Q+S\EBM.
\end{align*}
From well-posedness, it holds that $P$ is invertible.
Hence, the rank of $M$ is the same as the rank of $\Pi$ plus the dimension of
$P$, which is $n$.
This is a generalization of the formula in (5) in
\cite{wou91} to the direct term
case. In that paper, it is shown for real-valued structured\footnote{
Notice that we also have structured matrices, since the transfer function matrices are often sparse.
} matrices
$(A,B,C)$ and $s\in\complex$, how the generic rank  of the transfer function
$C(sI-A)^{-1}B$ can be given a graph-theoretic characterization.
From algebraic geometry, e.g. \cite{CoxLittleOShea1997},
it follows that the proof of Theorem~1 in  \cite{wou91} still holds
true in our setting, i.e.
$$\left\{ M\mid \rank \Pi<
\rank_{g} \Pi \right\}$$
is contained in a proper affine variety of the set of all $M$. Hence it
holds that $\rank_{g}  \Pi=\rank  \Pi$ for
almost all
$M$. With the diagonal arrangement above, the proof of Theorem~2 in
\cite{wou91} extends to our setting: the diagonal entries of $P$ are retained
when off-path edges are removed. We are therefore able to derive a graph-theoretic
characterization of the generic rank of $\Pi$.

To this end, let us define a directed graph $\mathcal G=(\mathcal V,\mathcal E)$ with vertex set
$\mathcal V$ and edge set $\mathcal E$, where $\mathcal V=\mathcal U\cup \mathcal W\cup \mathcal Z$
with $\mathcal U=\{\mathcal U_1,\ldots, \mathcal U_{m}\}$, $\mathcal W=
\{\mathcal W_1,\ldots,\mathcal W_{n}\}$, $\mathcal Z=
\{\mathcal Z_1,\ldots,\mathcal Z_p\}$.
We can now
associate each row and column in $M$ with elements in $\mathcal V$. The first $n$
columns are
labeled $\mathcal W_1$ to $\mathcal W_{n}$, and the next $m$ columns
$\mathcal U_1$ to $\mathcal U_m$. We proceed similarly with the
rows. We then
define the edges in $\mathcal E$ as all the ordered two-tuples $(v_i,v_j)$ with
$v_i,v_j\in\mathcal V$ such that $M_{v_j,v_i}$ is not a structural zero.
We do however, not include any self-loop edges, i.e.\ edges for which $v_i=v_j$.
If there are vertices $v_1,\ldots, v_k$ in $\mathcal V$ such that $(v_i,v_{i+1})\in \mathcal E$
for $i=1,\ldots, k-1$, we say that there is a path in $\mathcal G$ from $v_1$ to
$v_k$. If  $v_1\in \mathcal U$ and $v_k\in \mathcal Z$ we say that there
is a path from $\mathcal U$ to $\mathcal Z$. We say that two paths from
$\mathcal U$ to $\mathcal Z$ are {\it disjoint} if they have no vertex in common.
We say that an $l$-tuple of paths from $\mathcal U$ to $\mathcal Z$ is disjoint if each pair of paths in the $l$-tuple is disjoint.

We are now able to state the following result, which is a generalization of
Theorem~2 in \cite{wou91}.
\begin{theorem}
The maximum number of disjoint paths
in $\mathcal G$ from $\mathcal U$ to $\mathcal Z$ is equal to the generic rank of $\Pi$.
\end{theorem}
\begin{remark}
The problem of finding the
maximum number of disjoint paths can be formulated as a maximum flow
problem for an associated flow graph.
\end{remark}
\begin{remark}
For the above proof to hold, we need to assume that all non-zero entries of $M$ are free variables. This assumption is not
needed for the Gröbner basis condition to follow.
\end{remark}
\subsection{Gröbner Basis Condition}
We will now 
investigate the generic rank of $\Pi$ using Gröbner basis.
We are still considering $\Pi$ to be a rational function of the free entries of
$M$. We let $D_i$ be the least common multiple
of the denominators of the $i$th row of $\Pi$. Multiply
each row of $\Pi$ with this polynomial to obtain the polynomial
matrix $\tilde{\Pi}$. Then the generic rank of $\Pi$, when 
well-posedness holds, is the same as the generic rank of 
$\tilde{\Pi}$ when $D_i\neq 0$. Now, define the 
ideal $\mathcal I_k$ generated by all the minors of $\tilde{\Pi}$ of 
dimension $k$ together with the polynomial $1-tD$, where 
$D=\det(P)\prod_i D_i$ and $t$ is an additional variable.
Any additional denominator polynomials required by rational constraints or
parameterizations are also included as factors in $D$. Notice
that the equation $1-t D=0$ always has a solution where
$D\neq 0$ and $t=1/D$. All rank statements and solution sets below are
understood on the admissible domain $D\neq 0$.
From this the following theorem follows.
\begin{theorem}
Given the ideal $\mathcal I_k$ defined above it holds:
\begin{enumerate}
    \item If the reduced Gröbner basis of $\mathcal I_k$ is one, then 
    $\rank \Pi\geq k$.
    \item If the reduced Gröbner basis of $\mathcal I_k$, excluding the basis depending on $t$,  is zero, then 
    $\rank\Pi<k$.
    \item Otherwise $\rank_g\Pi\geq k$. The polynomials in the Gröbner basis
    that do not involve $t$, together with the restriction $D\neq 0$, define
    the set of entries of $M$ for which $\rank\Pi<k$. In particular, choosing
    $k=\rank_g\Pi$ gives the exceptional set where the rank is below its
    generic value.
\end{enumerate}
\end{theorem}
Without the restriction $D\neq 0$, the elimination ideal defines the algebraic
closure of the admissible set where $\rank\Pi<k$.
\begin{remark}
The lexicographic ordering should be used when the Gröbner basis is computed, and the variable
$t$ should be the largest variable according to that ordering. 
\end{remark}
\begin{remark}
Notice how known entries of $M$ are seamlessly treated since we work over the field of
rational functions of transfer functions.
\end{remark}
\begin{remark}
In case some of the entries of $M$ are constrained by rational equations,
we first clear their denominators and add the resulting polynomial equations
to the ideal, including those denominators as factors in $D$.
With algebraic constraints, the zero-ideal test in the theorem is interpreted
in the coordinate ring of the relevant reduced irreducible component.
Reducible constraint sets are treated componentwise, considering only
components that meet the admissible domain $D\neq 0$.
Specifically, one may consider constraining some entries to be the same.
However, then it is easier to just remove the duplicate variable to start with. 
\end{remark}
Notice that the generalizations in the remarks above are trivial to include
in the Gröbner basis condition. This is not so easy to do in the 
graphical condition of the previous subsection, which assumes all non-zero transfer functions to be free variables. 
\section{Identifiability}\label{sec:ident}
As was discussed in Section~\ref{sec:innovation} the closed loop 
transfer function can be recovered from the predictor using 
\eqref{eqn:predtogc}, which we can uniquely do under the assumption of
informativity. 
The next question that arises is whether we can recover parts of $M$ from $G_c$. From
\eqref{eqn:closed-tf}, it holds that 
$$G_c=R_yP^{-1}Q_r+S_{yr}$$
which is a rational expression  in the free entries of 
$(P,Q_r,R_y,S_{yr})$, which we will denote by the $k$-dimensional vector $X$. 
Hence, a relevant question is when 
a rational equation has a unique solution. 
However, in many cases, this is too 
much to ask for. What is a more relevant question is whether there is a
finite number of isolated solutions for a given $G_c$.  
To investigate this question, let $F$ be a function defined by the rule
$$F(X)=R_yP^{-1}Q_r+S_{yr}.$$
We keep $z$ symbolic and regard $F$ as a rational map over the field
$K=\mathbb C(z)$. Known rational transfer functions are included in its
coefficients. For the algebraic analysis, we consider varieties and fibers over
an algebraic closure of $K$. We let $V_o$ be the open subset of affine
$k$-space on which $P$ is invertible and all expressions defining $F$ are
well-defined. Admissible rational transfer-function networks correspond to
points of $V_o$ whose coordinates belong to $K$. For these points, equality
under $F$ means equality of rational functions, not equality at a fixed frequency.
We then define the set of {\it admissible} $G_c$ as $V_c=F(V_o)$. 
If we take $V_c$ to be the co-domain of $F$, it follows that $F$ is a surjective
function. The image $V_c$ need not itself be an affine variety. We use
$\dim(V_c)$ to denote the dimension of its algebraic closure, i.e.\ the smallest
affine variety containing $V_c$. Here and below, a property holds generically
if it holds outside a proper algebraic subset of the relevant irreducible
variety. We assume that $V_o$ is irreducible; this holds for the unconstrained
open subset of affine $k$-space considered above.

Unique solution and finite number of isolated solutions parallels the discussion
in \cite{9303890} regarding {\it generic network identifiability} and
{\it generic local network
identifiability}, where it is argued that the latter is easier to investigate than 
the former. There it is also conjectured 
that generic local network identifiability implies
generic network identifiability. 
\begin{definition}
We say that we have generic local network identifiability in the algebraic
sense if, for generic $G_c\in V_c$, the fiber defined by
\BEQ\label{eqn:open-closed}
F(X)=G_c,\qquad X\in V_o,
\EEQ
has dimension zero. We use this algebraic meaning of generic local network
identifiability throughout the following sections.

\end{definition}
\begin{remark}
A zero-dimensional fiber consists of finitely many isolated algebraic
solutions. Thus, generic local network identifiability allows a finite
ambiguity, whereas generic network identifiability requires uniqueness.
In particular, there are only finitely many admissible rational
transfer-function solutions. The algebraic solutions need not all be rational
transfer functions, and any additional requirements such as properness or
stability must also be checked when interpreting an alternative solution as an
admissible network. Legat and Hendrickx \cite{9303890} formulate local
identifiability using complex transfer-function values at a single frequency
$z$ and Euclidean neighborhoods. Here, we instead keep $z$ symbolic and use
algebraic fibers, without requiring a norm on transfer functions.

\end{remark}
Then we have the following theorem. 
\begin{theorem}
Under the irreducibility assumption on $V_o$, $X$ is generically
locally network identifiable in the algebraic sense from
\eqref{eqn:open-closed} if and only if 
$$\dim(V_o)=\dim(V_c).$$
\end{theorem}
\begin{proof}
For a fixed $G_c\in V_c$ define the  fiber
$$F^{-1}( G_c)=\left\{X\in V_o
\mid F( X) =G_c\right\}.$$
It holds that generic local network identifiability is equivalent to that
$\dim\left(F^{-1}( G_c)\right)=0$ for generic $G_c\in V_c$.
Since $F$ has dense image in the algebraic closure of $V_c$,
it follows from the fiber dimension theorem in algebraic geometry that
$$\dim(V_o)=\dim(V_c)+ \dim\left(F^{-1}(G_c)\right)$$
for generic $G_c\in V_c$. The result is immediate from this.
\end{proof}
\begin{remark}
In the unconstrained case, $V_o$ is a dense open subset of affine $k$-space,
so $\dim(V_o)=k$.
We define the polynomial matrices 
$S$ and $T$ such that $T\odot F(X)=S$ where $\odot$ is the Hadamard product of
matrices. We let $D$ be the least common multiple of the entries of $T$,
enlarged to include the polynomial factors excluded from $V_o$.
Then we define the polynomial matrix
$T\odot G_c-S$, and the ideal generated by the entries of this matrix and $1-tD$. 
We then compute the Gröbner basis for this ideal using the lexicographic ordering, where
the variables in $G_c$ should be the smallest variables according to this ordering. 
The part of the basis containing only the variables in $G_c$ defines the 
smallest variety that contains $V_c$,
c.f. \cite[Ch.3.3]{cox+lit+she92}. The dimension of this variety is the
same as the dimension of $V_c$. 
Then the algorithm in \cite{udd13} can be used to compute the dimension. 
\end{remark}
\begin{remark}
In case some of the entries of $X$ are known, we just exclude them when defining $V_o$ and the variable $X$.
\end{remark}
\begin{remark}
It is possible to put rational constraints on $X$. We then use the
constrained admissible variety as $V_o$ and its actual dimension. If this
variety is reducible, the dimension criterion is applied to each irreducible
component and the closure of its image under $F$.
In case the constraint is that one of the entries is the same as another entry, then 
it is easier to just replace one variable with the other to remove the constraint
We will discuss this in an example later on in more detail. 
\end{remark}
\begin{remark}
Notice that it is possible to discuss identifiability also when not all of
the predictor is estimated uniquely. Then we just consider the closed loop 
transfer function corresponding to the part of the predictor we estimate
uniquely. Then there will be fewer equations that determine the open loop 
transfer functions, but there could still be enough equations to guarantee generic
network identifiability. We will discuss this in more detail for an 
example later on. 
\end{remark}
What has been presented above
is a generalization of the results in \cite{hen+gev+baz19} to the
case when $Q_r\neq I$. It is also a generalization of the results in
\cite{SHI2022110093} to the case when $R_y\neq I$. 
Our 
result also generalizes the results in \cite{leg25}, where it is assumed that 
$Q_r$ and $R_y$ are given 
zero-one matrices. We also consider a direct term $S_{yr}$, which other work
do not. 

Since $K=\mathbb C(z)$ has characteristic zero, generic local network
identifiability in the algebraic sense is equivalent, in the unconstrained
case, to the Jacobian of the entries of $F$ with respect to $X$ having rank
$k$ over $K(X)$. With constraints, the differential is restricted to the
tangent space of $V_o$ at a generic smooth point and must have rank
$\dim(V_o)$. Here, we focus on a Gröbner basis formulation.
\section{Sub-Networks}
So far we have discussed the identification of a whole network. Often one
is only interested in parts of a network. Without loss of generality we will
assume that the part of the network we are interested in corresponds to the 
leading entries $w_k^A$ of $w_k=(w^A_k,w^B_k)$. We then partition $M$ as
$$
\begin{aligned}
M&=\BBM P&Q_r&Q_e\\-R_y&S_{yr}&S_{ye}\\0&I&0\EBM
&=\BBM P_A&P_{AB}&Q_{Ar}&Q_{Ae}\\
P_{BA}&P_{B}&Q_{Br}&Q_{Be}\\
-R_{yA}&-R_{yB}&S_{yr}&S_{ye}\\0&0&I&0\EBM.
\end{aligned}
$$
We have
\begin{align*}
P_Aw^A_k&=\BBM -P_{AB}&Q_{Ar}\EBM\BBM w^B_k\\r_k\EBM+Q_{Ae}e_k\\
y_k&=R_{yA}w^A_k+\BBM R_{yB}&S_{yr}\EBM\BBM w^B_k\\r_k\EBM+S_{ye}e_k.
\end{align*}
Hence, we can investigate generic local network identifiability of 
network A by making the following replacements
in our previous definitions: 
\begin{align*}
P_A&\rightarrow P\\
\BBM -P_{AB}&Q_{Ar}\EBM&\rightarrow Q_r\\
R_{yA}&\rightarrow R_y\\
\BBM R_{yB}&S_{yr}\EBM&\rightarrow S_{yr}\\
w^A_k&\rightarrow w_k\\
(w^B_k,r_k)&\rightarrow r_k.
\end{align*}
We then obtain the following rational equation for investigating identifiability:
$$R_{yA} P_A^{-1}\BBM-P_{AB}&Q_{Ar}\EBM+\BBM R_{yB}&S_{yr}\EBM=G_c.$$
If there are columns in 
$$\BBM -P_{AB}&Q_{Ar}\\R_{yB}&S_{yr}\EBM$$
that are zero, then they may be removed from the rational function. Similarly any zero rows in $\BBM R_{yA}&S_{yr}\EBM$
may be removed.
When it comes to informativity, we realize that the B-part of the network 
provides excitation of the A-part of the network, and hence the whole 
network should be considered. Since $w_k^B$ is part of the 
exogenous input, we should consider 
$$
\begin{aligned}
M&=\BBM P_A&P_{AB}&Q_{Ar}&Q_{Ae}\\
P_{BA}&P_{B}&Q_{Br}&Q_{Be}\\
-R_{yA}&-R_{yB}&S_{yr}&S_{ye}\\0&-I&0&0\\0&0&I&0\EBM
\end{aligned}
$$
when we investigate informativity. However, rows of the second last
block row that correspond to the zero columns in $\BBM -P_{AB}\\R_{yB}\EBM$ discussed above should be removed.

As seen above, the exogenous signal contains also measurements of the 
node signals in the B-part of the network. We also need to estimate models
for parts of the B-network. Because of sparsity this might not be a big 
problem, since we may not have to measure all of $w^B_k$. However, this is in many cases still not desirable. 
If we instead of measuring
the node signals $w^B_k$ can measure 
$$\BBM \tilde w_k^B\\\hat w_k^B\EBM=\BBM -P_{AB}\\R_{yB}\EBM w^B_k,$$
we will not need to estimate models for any parts of the B-network. 
Then we obtain the equations:
\begin{align*}
P_Aw^A_k&=\BBM I&0& Q_{Ar}\EBM\BBM \tilde w^B_k\\\hat w_k^B\\r_k\EBM+Q_{Ae}e_k\\
y_k&=R_{yA}w^A_k+\BBM 0&I&S_{yr}\EBM\BBM \tilde w^B_k\\\hat w_k^B\\r_k\EBM+S_{ye}e_k.
\end{align*}
The replacement will instead be:
\begin{align*}
P_A&\rightarrow P\\
\BBM I&0&Q_{Ar}\EBM&\rightarrow Q_r\\
Q_{Ae}&\rightarrow Q_e\\
R_{yA}&\rightarrow R_y\\
\BBM 0&I&S_{yr}\EBM&\rightarrow S_{yr}\\
S_{ye}&\rightarrow S_{ye}\\
w^A_k&\rightarrow w_k\\
(\tilde w^B_k,\hat w^B_k,r_k)&\rightarrow r_k.
\end{align*}
We then obtain the following rational equation for investigating identifiability:
$$R_{yA} P_A^{-1}\BBM I&0&Q_{Ar}\EBM+\BBM 0&I&S_{yr}\EBM=G_c$$
which now contains dynamics of only the A-part of the network.
Any zero rows in $\BBM R_{yA}&S_{yr}\EBM$
may be removed.
If there are zero rows in 
$$\BBM -P_{AB}\\R_{yB}\EBM$$
it follows that there are zero signals in $(\tilde w_k^B,\hat w_k^B)$. The 
corresponding columns in the above rational matrix may then be removed. 
Also any zero columns in $\BBM Q_{Ar}\\S_{yr}\EBM$ may be removed. 
Since $(\tilde w_k^B,\hat w_k^B)$ is part of the 
exogenous input, we should consider 
$$
\begin{aligned}
M&=\BBM P_A&P_{AB}&Q_{Ar}&Q_{Ae}\\
P_{BA}&P_{B}&Q_{Br}&Q_{Be}\\
-R_{yA}&-R_{yB}&S_{yr}&S_{ye}\\0&P_{AB}&0&0\\0&-R_{yB}&0&0\\0&0&I&0\EBM
\end{aligned}
$$
when we investigate informativity. However, any zero rows in the fourth and
fifth block rows should be removed, since they correspond to zero signals in $(\tilde w_k^B,\hat w_k^B)$. 

Notice that the criterion for 
informativity holds for any pattern of the matrices defining the $M$-matrix. 
Based on this, we now realize that 
we are able to analyze informativity and generic local network identifiability
for sub-networks using the methods we have developed for a whole network. 
It is possible to combine the two above methods
depending on what can be measured so that only some of the columns of $P_{AB}$ and
$R_{yB}$ are estimated. 
\section{Example}\label{sec:discussion}
In this section, we will investigate the theory we have developed on some examples. 
\begin{exmp}
We will look at the very simple example in 
Figure~\ref{fig:simple} to illustrate why the 
conditions we have presented are more general than the ones 
previously presented in the literature. 
\begin{figure}[htbp]
\begin{center}
\resizebox{0.35\textwidth}{!}{\begin{tikzpicture}[auto, node distance=2cm]

\node [input, name=r1] {};
\node [sum, right of=r1, node distance=1.5cm] (sum1) {$+$};
\node [block, right of=sum1, node distance=1.5cm] (G1) {$ G^1$};
\node [sum, right of=G1, node distance=1.5cm] (sum2) {$+$};
\node [input, above of=sum2, node distance=1.5cm] (v1) {};

\node [block, below of=G1, node distance=1.5cm] (G2) {$ G^2$};
\node [sum, below of=sum2, node distance=1.5cm] (sum3) {$+$};
\node [sum, below of=sum1, node distance=1.5cm] (sum4) {$+$};
\node [input, below of=sum4, node distance=1.5cm] (v2) {};
\node [input, right of=sum3, node distance=1.5cm] (r2) {};

 \draw [-Latex] (r1) -- node {$ r^1$} (sum1);  
\draw [-Latex] (sum1) -- node {$ w^2$} (G1); 
\draw [-Latex] (G1) -- node {} (sum2);
\draw [-Latex] (v1) -- node {$ v^1$} (sum2);

 \draw [-Latex] (r2) -- node {$ r^2$} (sum3);  
\draw [-Latex] (sum3) -- node {$ w^4$} (G2); 
\draw [-Latex] (G2) -- node {} (sum4);
\draw [-Latex] (v2) -- node {$ v^2$} (sum4);

\draw [-Latex] (sum4) -- node {$ w^3$} (sum1);
\draw [-Latex] (sum2) -- node {$w^1$} (sum3);

\end{tikzpicture}}
\end{center}
\caption{Block diagram for a simple network.}
\centering
\label{fig:simple}
\end{figure}
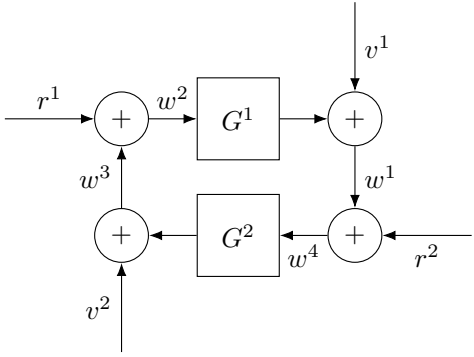
We assume that $v^1_k=H^1 e^1_k$, $v^2_k= H^2 e^2_k$ and that
$r^1_k= Q^1 r_k$ and $r^2_k= Q^2 r_k$. We assume that $H^1\not\equiv0$
and $H^2\not\equiv0$. Notice that we have
the same single $r_k$ entering the loop in two places.
We are interested in estimating
$( G^1, Q^1)$, and use a predictor for
predicting $w^1_k$. We can measure $(w^3_k,r_k)$ in 
addition to $w^1_k$, and therefore the predictor input is $z_k=(w^1_k,w^3_k,r_k)$.
We obtain the 
following equations
\BEQ\label{eqn:example}
\begin{aligned}
\BBM w^1_k\\w^2_k\\w^3_k\\w^4_k\EBM&=
\BBM 0& G^1&0&0\\0&0&1&0\\0&0&0& G^2\\
1&0&0&0\EBM
\BBM w^1_k\\w^2_k\\w^3_k\\w^4_k\EBM+
\BBM 0&H^1&0\\ Q^1&0&0\\0&0&H^2\\Q^2&0&0\EBM \BBM r_k\\e^1_k\\e^2 _k\EBM\\
z_k&=\BBM 1&0&0&0\\0&0&1&0\\0&0&0&0\EBM
\BBM w^1_k\\w^2_k\\w^3_k\\w^4_k\EBM+\BBM 0\\0\\1\EBM r_k.
\end{aligned}
\EEQ
It is straightforward to see that there are three disjoint paths:
$(\mathcal E_1,\mathcal W_1,\mathcal Z_1)$, $(\mathcal E_2,\mathcal W_3,\mathcal Z_2)$ and $(\mathcal R_1,\mathcal Z_3)$, and hence we have
informativity under the standing well-posedness and stability assumptions,
assuming that the spectrum for $(r_k,e_k^1,e_k^2)$ is positive definite for
almost all frequencies. We can also verify this by directly looking at the
transfer function 
$$\Pi =\BBM\frac{ G^1( Q^1+ G^2 Q^2)}{1- G^1 G^2}&\frac{H^1}{1- G^1 G^2}
&\frac{ G^1H^2}{1- G^1 G^2}\\
\frac{ G^2( G^1 Q^1+ Q^2)}{1- G^1 G^2}&\frac{ G^2H^1}{1- G^1 G^2}
&\frac{H^2}{1- G^1 G^2}\\
1& 0&0\EBM$$
which has determinant $H^1H^2/(1- G^1 G^2)$. Hence, both the rank
and the generic rank are full under well-posedness and the nonzero noise-filter
assumptions.

The closed loop system transfer function from $(w^3_k,r_k)$ to $w^1_k$ is
given by
$$\setlength{\arraycolsep}{3.5pt}\BBM  G_c^1& G_c^2\EBM=\BBM 1&0\EBM\BBM 1&- G^1
\\0&1\EBM^{-1}\BBM 0&0\\1& Q^1\EBM
=\BBM  G^1& G^1Q^1\EBM.$$
The dimensions of the domain and the image closure are $\dim V_c=\dim V_o=2$ and we conclude that
$$( G^1, Q^1)=
( G_c^1, G_c^2/ G_c^1),\qquad G_c^1\not\equiv0.$$
Here, nonzero means not identically zero as a rational function.
Hence, we have generic network identifiability, and therefore also generic
local network identifiability. Uniqueness fails at $G_c=(0,0)$, where $G^1=0$
and $Q^1$ is arbitrary. Thus, identifiability does not hold for every admissible
closed loop transfer function.
\end{exmp}

\begin{exmp}
Let us now consider the case where we are only interested in estimating
$G^1$ and $Q^1$ is known. Using the same measured signals, we can form
$\xi_k=w_k^3+Q^1(q)r_k$. The network equation becomes
$$w_k^1=G^1(q)\xi_k+H^1(q)e_k^1.$$
Under the assumed innovation representation, the predictor can therefore be
written as
$$\hat w_k^1=W_y(q)w_k^1+V(q)\xi_k,\qquad G^1=(1-W_y)^{-1}V.$$
Equivalently, the original predictor columns are $(W_y,V,VQ^1)$. Following
Remark~\ref{rem:predictor-reduction}, it suffices to establish informativity of
$(w_k^1,\xi_k)$, rather than positive definiteness of the spectrum of all three
original predictor inputs.

With $e^2=0$, the transfer matrix from $(r,e^1)$ to $(w^1,\xi)$ has determinant
$$-\frac{H^1(Q^1+G^2Q^2)}{1-G^1G^2}.$$
With $r=0$, the corresponding determinant for inputs $(e^1,e^2)$ is
$$\frac{H^1H^2}{1-G^1G^2}.$$
In either case, a positive-definite spectrum of the active inputs for almost all
frequencies ensures informativity when the displayed determinant is not
identically zero, under the standing well-posedness and stability assumptions.
Hence, known $Q^1$ allows reduced excitation by combining predictor inputs,
rather than simply discarding an unknown predictor column. This illustrates
how prior knowledge can reduce the excitation needed to identify $G^1$.
\end{exmp}

\begin{exmp}
We consider again the network in Figure~\ref{fig:simple}, but this time
we are interested in estimating both transfer functions $( G^1,
 G^2)$. For simplicity, we consider noise-free experiments under the standing
well-posedness and stability assumptions. We perform two separate experiments
where we first excite with only $r_k^1$ and measure $w^1_k$, and then we excite
with only $r_k^2$ and measure $w^3_k$. With the spectrum of each active input
positive for almost all frequencies, the corresponding closed loop transfer
function can be determined directly from the input-output relation. This does
not require unique determination of an unrestricted predictor. We have
$$w_k^1= G_c^1 r_k^1,\quad w_k^3= G_c^2 r_k^2$$
where
$$ G_c^1=\frac{ G^1}{1- G^1 G^2},
\quad  G_c^2=\frac{ G^2}{1- G^1 G^2}.$$
If we replace $( G^1, G^2)$
with $(-1/ G^2,-1/ G^1)$ we obtain the same 
$( G_c^1, G_c^2)$ as rational functions. Generically, these are the two
isolated solutions. Hence, this is an example where generic local
identifiability in the algebraic sense does not imply generic network identifiability.
Notice that
$$\BBM  w^1_k\\w^2_k\\w^3_k\\w^4_k\EBM=
\frac{1}{1- G^1 G^2}
\BBM  G^1& G^1 G^2\\
1& G^2\\
 G^1 G^2& G^2\\
 G^1& 1
\EBM\BBM r^1_k\\r^2_k\EBM$$
and it is not possible to pick out these two closed loop transfer functions
by multiplying from the left and the right with zero-one matrices. 
However, if we duplicate the network and constrain the transfer functions
in the two networks to be the same, it is possible to pick out them
with zero-one matrices. Thus, the conjectured implication from local
identifiability to identifiability in \cite{9303890} does not extend in general
to networks with equality constraints between repeated transfer functions.
\end{exmp}
\section{Numerical Examples}
The above examples were possible to analyze without using any of the theory
we have developed. In this section, we will show how a Matlab implementation 
based on Gröbner basis calculations can be used to analyze fairly large examples,
that are not so easy to analyze by hand. 
\begin{exmp}
Consider the dynamical network defined by $(P,Q_r,R_y,S_{yr})$, where $P=I-G$, $S_{yr}=0$ and 
\begin{align*}
G&=\BBM 0&G^{12}& 0 & 0 & 0\\
        0& 0& G^{23} & 0 &0\\
        0&0&0&G^{34}&0\\
        0&0&0&0&G^{45}\\
        G^{51}&0&0&0&0\EBM\\
Q_r&=\BBM Q^1&1&0&0&0\\
        0&Q^2&1&0&0\\
        0&0&Q^3&1&0\\
        0&0&0&Q^4&1\\
        Q^5&0&0&0&1\EBM\\
R_y&=\BBM R^1& 0 & 0& 0&0\\
        0& 1&0&0&0\\
        0&0&0&1&0\EBM.
\end{align*}
For this example, we have $\dim V_o=11$, and using Matlab, we can compute $\dim V_c=11$ in 11 minutes on a laptop
computer. Hence, we have generic local network identifiability. In a second we can verify that
the generic row rank of $\Pi$ is full for the case when $Q_e=I$ and $S_{ye}=0$,
and therefore we also have informativity generically, under the standing
assumptions, if the spectrum of $(r_k,e_k)$ is positive definite almost everywhere.
\end{exmp}
\begin{exmp}
We consider the same dynamic network as in the previous example, but now we are only interested in identifying
the transfer functions $(G^{12},Q^1,Q^2,R^1)$. We let $y_k=(R^1(q)w_k^1,w_k^2)$
and we measure $\tilde w_k^B=G^{23}w^3_k$ from the B-part of the network. We then have that
$$F(X)=\BBM R^1&0\\0&1\EBM\BBM 1&-G^{12}\\0&1\EBM^{-1}\BBM 0&Q^1&1&0\\1&0&Q^2&1\EBM$$
with the input $(\tilde w_k^B,r_k^1,r^2_k,r_k^3)$ and the output $y_k$. The $M$-matrix for investigating
informativity is
$$M=
\BBM
P&Q_r&Q_e\\
-R_y&S_{yr}&S_{ye}\\
-\tilde P&0&0\\
0&I&0\EBM$$
where $P$ and $Q_r$ are as in the previous example, $Q_e=I$, $S_{yr}=0$, $S_{ye}=0$,and 
$$R_y=\BBM R^1& 0 & 0& 0&0\\
        0& 1&0&0&0\EBM,\quad \tilde P=\BBM 0&0&G^{23}&0&0\EBM.$$
Notice that the zero matrices have one fewer row as compared to the previous example. In about a second we can
verify that we have generic local network identifiability and generic informativity.
Here and in the reduced-excitation cases below, informativity is understood under
the standing assumptions and a positive-definite spectrum of the remaining active
external signals for almost all frequencies.
We can also investigate if $(r_k^4,r_k^5)$ can be equal to zero. We remove the
columns corresponding to $r_k^4,r_k^5$ and the corresponding measurement rows
in the bottom block row of $M$, and we obtain that we still have generic
informativity. It is also possible to investigate if we in addition can have
$r_k^3$ equal to zero. We then remove its input column and corresponding
measurement row from $M$. We also investigate generic local network
identifiability after removing the last column in $F(X)$. It turns out that
we have both generic informativity and generic local network identifiability
also for this case.

\end{exmp}
\section{Conclusions}
In this paper, we have shown how to investigate informativity and 
generic local network identifiability in the algebraic sense for networks of dynamical systems.
In the free-entry setting, our sufficient conditions for generic informativity
are characterized in terms of vertex-disjoint paths of a graph 
associated with the dynamic network. Informativity can also be cast as an
algebraic geometry problem where Gröbner bases are computed. 
The generic local network identifiability
result is phrased as the admissible domain and the closure of its image having the
same dimension. This can also be investigated using Gröbner bases. 
Our approach seamlessly generalizes to the case where some transfer
functions are known. This is in contrast to the work presented in
e.g. \cite{SHI2022110093}. Actually, our framework takes advantage of known 
transfer functions in the computations of the Gröbner basis in that
computation time can be reduced.
Future work will focus on investigating the interplay
between informativity and generic local network identifiability in more detail.

\bibliographystyle{IEEEtran}        

%
\end{document}